\documentclass[11pt]{article}
\usepackage{amsmath,amsthm}
\usepackage{mathtools}
\mathtoolsset{showonlyrefs}

\usepackage{ifxetex,ifluatex}
\ifxetex
\usepackage{xltxtra}
\else
\ifluatex
\usepackage{luatextra}
\else
\usepackage{amssymb}
\usepackage[utf8]{inputenc}
\usepackage[T1]{fontenc}
\usepackage{wasysym}
\usepackage{mathpazo,tgcursor,tgpagella}
\usepackage{textcomp}
\csname else\expandafter\endcsname
\romannumeral -`0
\fi
\fi
\iftrue\relax%
\defaultfontfeatures{Ligatures=TeX}
\usepackage[math-style=TeX, vargreek-shape=TeX]{unicode-math}
\AtBeginDocument{%
  
}
\fi

\usepackage[footnotesize,nooneline,bf]{caption}

\usepackage[letterpaper,margin=1in]{geometry}
\usepackage{graphicx,authblk}
\usepackage{color}
\usepackage{tikz}
\usetikzlibrary{positioning, decorations.markings}
\usetikzlibrary{plotmarks}
\usepackage[authoryear]{natbib}
\usepackage[compact]{titlesec} 
\usepackage{paralist}
\usepackage{enumitem}
\setenumerate[1]{label=(\roman*)}
\usepackage{xparse}
\usepackage{ifthen}
\usepackage{xspace}

\usepackage[bookmarksnumbered,bookmarksopen,unicode]{hyperref}
\DeclareMathOperator{\STAB}{STAB}
\newcommand{\STABu}{\STAB\sp{u}}
\newcommand{\STABnu}{\STAB\sp{nu}}
\DeclareMathOperator{\rank}{rank}
\newcommand{\nnegrk}{\rank_{+}} 

\NewDocumentCommand{\UDISJ}{o}{\ensuremath{\mathrm{UDISJ}%
  \IfValueT{#1}{(#1)}}\xspace}
\newcommand{\RR}{\mathbb{R}}
\newcommand{\N}{\mathbb{N}} 
\newcommand{\R}{\RR}

\DeclareMathOperator{\xc}{fc}

\newcommand{\inp}[2]{\langle #1,#2 \rangle}

\newcommand{\polylog}{\operatorname{polylog}}
\newcommand {\card}[1]{\left|#1\right|}

\newcommand*{\set}[2]{\left\{#1\,\middle|\,#2\right\}}

\NewDocumentCommand{\conv}{mo}{\operatorname{conv}\left(#1%
    \IfValueT{#2}{\,\middle|\,#2}%
  \right)}

\NewDocumentCommand{\probability}{om}{\operatorname{\mathbb{P}}\left[#2
    \IfValueT{#1}{\,\middle|\,#1}\right]}
\NewDocumentCommand{\expectation}{om}{\operatorname{\mathbb{E}}\left[#2
    \IfValueT{#1}{\,\middle|\,#1}\right]}
\newcommand*{\variance}[1]{\operatorname{Var}\left[#1\right]}
\newcommand*{\covariance}[2]{\operatorname{Cov}\left[#1,#2\right]}
\newcommand*{\bentropy}[1]{\mathbb{H}\left[#1\right]}

\newcommand{\indsubgraph}{\overset{\textup{ind}}{\subseteq}}
\newcommand{\notindsubgraph}{\overset{\textup{ind}}{\nsubseteq}}

\newcommand{\smallgadget}[1]{#1\sp{\circ}}

\DeclarePairedDelimiter{\roundup}{\lceil}{\rceil}

\tikzset{
  plotlabel options/.style={above},
  plotlabel/.style 2 args={%
    postaction={ decorate,
      decoration={ markings, mark=at position #1 with
        \node[plotlabel options]{#2};}}}}

\DeclareMathOperator{\Aut}{Aut}
\newcommand{\avgd}{\bar{d}}

\newcommand{\gadget}[1]{#1^{\sim}}

\newcommand*{\size}[1]{\left|#1\right|}
\newtheorem{thm}{Theorem}[section]
\newtheorem{mainthm}[thm]{Main Theorem}
\newtheorem{lem}[thm]{Lemma}

\newtheorem{cor}[thm]{Corollary}
\newtheorem{rem}[thm]{Remark}
\newtheorem{conjecture}[thm]{Conjecture}
\theoremstyle{definition}
\newtheorem{example}[thm]{Example}
\newtheorem{defn}[thm]{Definition}

\hypersetup{
  pdftitle={Average case polyhedral complexity of the maximum stable set problem},
  pdfauthor={Gábor Braun, Samuel Fiorini,
    Sebastian Pokutta},
  pdfkeywords={polyhedral approximation, extended formulation,
    stable sets},
  pdfsubject={MSC2000 Primary
    68Q17
    ;
    Secondary
    05C69, 
    05C80, 
    90C05.
}}

\title{Average case polyhedral complexity\\ of the maximum stable set problem}
\date{March 2, 2016}
\author[1]{Gábor Braun}
\affil[1]{ISyE, Georgia Institute of Technology,
  Atlanta (GA),
  USA.
  \textit{Email:}~gabor.braun@isye.gatech.edu}

\author[2]{Samuel Fiorini}
\affil[2]{Department of Mathematics,
  Université libre de Bruxelles CP 216,
  Bd. du Triomphe,
  1050 Brussels,
  Belgium.
  \textit{Email:}~sfiorini@ulb.ac.be}

\author[3]{Sebastian Pokutta}
\affil[3]{ISyE, Georgia Institute of Technology,
  Atlanta (GA),
  USA.
  \textit{Email:}~sebastian.pokutta@isye.gatech.edu}

\begin{document}
\maketitle

\begin{abstract}
We study the minimum number of constraints needed to formulate 
random instances of the maximum stable set problem via linear
programs (LPs), in two distinct models. 
In the uniform model, the constraints of the LP are not allowed
to depend on the input graph, which should be encoded solely in 
the objective function. There we prove a $2^{\Omega(n/\log n)}$ 
lower bound with probability at least $1 - 2^{-2^n}$ for every LP 
that is exact for a randomly selected set of instances; each graph 
on at most $n$ vertices being selected independently with probability 
$p \geqslant 2^{- \binom{n/4}{2} + n}$.
In the non-uniform model, the constraints of the LP may depend on
the input graph, but we allow weights on the vertices. The input
graph is sampled according to the $G(n,p)$ model. There we obtain
upper and lower bounds holding with high probability for various 
ranges of $p$. We obtain a super-polynomial lower bound all the 
way from $p = \Omega\left( \frac{\log^{6+\varepsilon} n}{n} \right)$
to $p = o\left( \frac{1}{\log n} \right)$.
Our upper bound is close to this as there is only
an essentially quadratic gap in the exponent, which currently also 
exists in the worst-case model. Finally, we state a conjecture that
would close this gap, both in the average-case and worst-case
models.
\end{abstract}

\section{Introduction}
\label{sec:introduction}

In the last four years, extended formulations have gained considerable
interest in various areas, including discrete mathematics, combinatorial
optimization, and theoretical computer science. 
The key idea underlying extended formulations is that
with the right choice of variables,
various
combinatorial optimization problems
can be \emph{efficiently} expressed via linear programs (LPs).
This asks for studying the intrinsic difficulty
of expressing optimization problems through a single
LP, in terms of the minimum number of necessary \emph{constraints}.
In turn, this leads to a complexity measure that we call loosely
here `polyhedral complexity' (precise definitions are given later in
Section~\ref{sec:preliminaries}).

On the one hand, there is an ever-expanding collection of examples of 
small size extended formulations. For instance, \cite{Williams01} has
expressed the minimum spanning tree problem on a planar graph with only
a linear number of (variables and) constraints, while in the natural 
edge variables the LP has an exponential number of constraints. There
exist numerous other examples, see e.g., the surveys by
\cite{ConfortiCornuejolsZambelli10} and \cite{Kaibel11}.

On the other hand, a recent series of breakthroughs in lower bounds
renewed interest for extended formulations~
\citep{Rothvoss13MPA,FMPTW15jour,bfps2012jour,braverman2012information,BP2013,CLRS13,Rothvoss14}.
These breakthroughs make it now conceivable to quantify the polyhedral
complexity of \emph{any} given combinatorial optimization problem
\emph{unconditionally}, that is, independently of conjectures such as
P $\neq$ NP, and without extra assumption on the structure of the LP.

Although a polynomial upper bound on the polyhedral complexity 
yields a polynomial upper bound on the true algorithmic complexity 
of the problem e.g., through interior point methods—provided that
the LP can be efficiently constructed—the
converse does not hold in general, as the following recent examples show.

\cite{CLRS13} proved
that every LP for MAXCUT with approximation factor
at most $2-\varepsilon$ needs at least
$n^{\Omega\left(\frac{\log n}{\log \log n} \right)}$ constraints,
while the approximation factor of the celebrated
SDP-based polynomial time algorithm of \cite{GoemansWilliamson95} is 
close to $1.13$.
\cite{Rothvoss14} showed
a $2^{\Omega(n)}$ lower bound on the size of any LP
expressing the perfect matching problem,
despite having a polynomial time algorithm by \cite{Edmonds65}.
\cite{BP2015matching} show that the matching polytope does not
admit any fully-polynomial size relaxation scheme (the polyhedral
equivalent of an FPTAS). 

In this paper, we consider the problem of determining the 
\emph{average case} polyhedral complexity of the maximum stable 
set problem, in two different models: `uniform'
and `non-uniform', see Section \ref{sec:contribution} below. 
Roughly, the uniform model asks for a single LP that
works for a given set of input graphs. In the non-uniform 
model the LP can depend on the input graph $G$ but should 
work for every choice of weights on the vertices of $G$
(in particular, for all induced subgraphs of $G$).

We show that the polyhedral complexity of the maximum stable set
problem remains high in each of these models, when the input graph
is sampled according to natural
distributions. Therefore, we conclude that \emph{the (polyhedral)
hardness of the maximum stable set problem is not concentrated on
a small mass of graphs but is spread out through all graphs}.

\subsection{Contribution}
\label{sec:contribution} 

We present the first strong and unconditional results on the 
average case size of LP formulations for the maximum stable 
set problem. In particular, we establish that the maximum stable
set problem in two natural average case models does 
not admit a polynomial size linear programming formulation, 
even in the unlikely case that P $=$ NP.

\paragraph{Uniform model}
In the \emph{uniform model} the feasible solutions
are \emph{independent} of the instances.  The instances will
be solely encoded into the objective functions. This ensures that no
complexity of the problem is leaked into an instance-specific
formulation. A good example of a uniform model is the TSP polytope 
over \(K_n\) with which we can test for Hamiltonian cycles in any
graph with \(n\) vertices by choosing an appropriate objective
function.

In the case of the maximum stable set problem, we consider a 
random collection of input graphs $G$, where each graph $G$ with 
$V(G) \subseteq [n]$ is included in the collection with probability 
\(p \geqslant 2^{- \binom{n/4}{2} + n}\) and show that with probability 
at least \(1-2^{-2^n}\), every LP for the maximum stable set problem on 
such collection of input graphs has at least \(2^{\Omega(n / \log n)}\) 
inequality constraints.

\paragraph{Non-uniform model}
In the \emph{non-uniform model} we consider
the stable set problem for a \emph{specific but random graph}.
The polyhedral description may depend heavily on the
chosen graph. We sample a graph \(G\) in the Erdős–Rényi \(G(n,p)\) 
model, i.e., \(G\) has \(n\) vertices, and every pair of vertices
is independently connected by an edge with probability \(p\).
We then analyze the stable set polytope \(\STAB(G)\) of \(G\).
If \(p\) is small enough, so that the obtained graph is sufficiently
sparse, it will contain an induced subgraph allowing a polyhedral 
reduction from the correlation polytope, and of sufficient size.
Via this reduction we derive strong lower bounds on the size
of any LP expressing \(\STAB(G)\) that hold with high probability.
In particular, we obtain superpolynomial lower bounds
for \(p\) ranging between $\Omega(\frac{\log^{6+\varepsilon} n}{n})$ 
and $o(\frac{1}{\log n})$. For example for \(p = n^{-\varepsilon}\) 
and \(\varepsilon < 1/4\), any LP has at least \(2^{\Omega(\sqrt{n^
\varepsilon \log n})}\) constraints w.h.p.\
(with high probability), and for \(p =
\Omega (\frac{\log^{6 + \varepsilon} n}{n}) \), any LP has at least 
\(n^{\Omega(\log^{\varepsilon/5} n)}\) constraints w.h.p. 
Figure~\ref{fig:results} illustrates our lower bounds. In the figure,
\(\xc(G(n,p))\) denotes the formulation complexity of the stable set
problem on $G \sim G(n,p)$, which is the minimum number of constraints
in an LP formulation of the problem, see below.

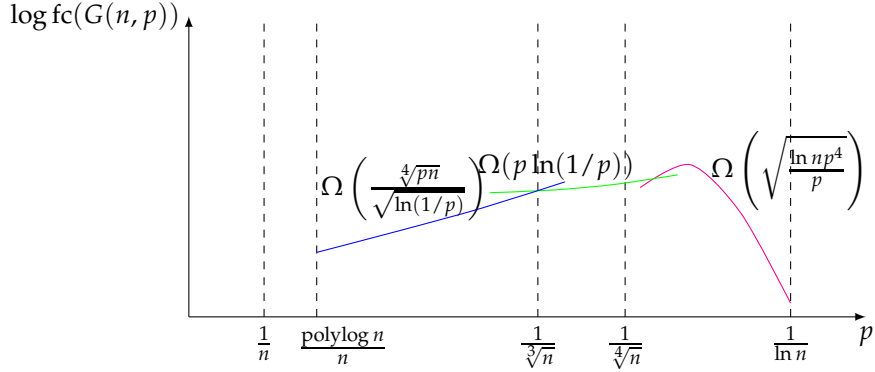
\begin{figure}[htb]
  \centering
\small
\begin{tikzpicture}[x=10cm, y=2cm,
  every plot/.style={smooth, samples=4}]
  \draw[-latex] (0,0) -- ++(up:2) coordinate (y-top)
  node[left]{\(\log \xc(G(n,p))\)};
  \draw[-latex] (0,0) -- ++(right:.9)
  node[below]{\(p\)};
  \draw[dashed]
  \foreach \x/\xtext in
    {0.1/n, 0.464/\sqrt[3]{n}, 0.58/\sqrt[4]{n},
    0.8/\ln n}
    {(\x,0) node[below, solid]{\(\frac{1}{\xtext}\)}
      edge (\tikztostart |- y-top)}
    (0.17,0) edge (\tikztostart |- y-top)
    node[below, xshift=1em, solid]{\(\frac{\polylog n}{n}\)}
    ;

  \draw[draw=magenta,
  plotlabel options/.style={right},
  plotlabel={.3}{\hyperref[eq:non-uni-xc-general]%
    {\(\Omega\left( \sqrt{\frac{\ln n p^{4}}{p}} \right)\)}}]
  plot[id=gr-fourth, domain=0.6:0.8]
  function{40*sqrt( log ((exp(1)*x/.65)**4) / x ) - 98.2};
  \draw[draw=green, plotlabel={.35}{\hyperref[eq:non-uni-xc-1/3-1/4]
    {\(\Omega(p \ln(1/p))\)}}]
  plot[id=third-fourth, domain=0.4:0.65]
  function{(x * log(1.0/x))**(-0.5) / 2};
  \draw[draw=blue, plotlabel={.35}{\hyperref[eq:non-uni-xc-less-1/3]
    {\(\Omega \left( \frac{\sqrt[4]{pn}}{\sqrt{\ln(1/p)}}
      \right)\)}}]
  plot[id=less-third, domain=0.17:0.5]
  function{sqrt( sqrt(x*10) / log(1.0/x)) / 2};
\end{tikzpicture}
  \caption{\label{fig:results}%
    Comparing lower bounds on \(\xc(G(n,p))\)
    for various regimes.
    For \(p\) close to \(1 / \sqrt[3]{n}\)
    the blue and green lines
    provide roughly the same bounds.
    For \(p\) significantly above \(1 / \sqrt[4]{n}\)
    the magenta line
    outperforms
    the green line.}
\end{figure} 

\subsection{Outline}
\label{sec:outline} 
In Section~\ref{sec:preliminaries} we recall basics on 
extended formulations. We introduce the model of general linear
programming formulations in Section~\ref{sec:unif-LP-form}.
We then establish bounds on the average case complexity for 
the uniform model of the maximum stable set problem in Section~\ref{sec:aver-case-compl-unif}. 
In Section~\ref{sec:aver-case-compl-nonunif} we consider the
non-uniform model and derive lower bounds as well as upper bounds. 
We conclude with a conjecture in Section~\ref{sec:concluding-remarks}.

\section{Preliminaries}
\label{sec:preliminaries}

We use $\log$ to denote the base-$2$ logarithm. Let $G$ be a graph. 
We denote by $\alpha(G)$ the maximum size of a stable set in $G$. This is
the \emph{stability number} of $G$. If $S$ is any set, we let \(G[S]\) 
denote the induced subgraph of \(G\) on \(V(G) \cap S\). For a positive
integer $n$ we let $[n] \coloneqq \{1,2,\ldots,n\}$.

\subsection{Linear programming formulations}
\label{sec:unif-LP-form}

The notion of linear programming formulation we use here 
generalizes both that appearing in the work of \cite{CLRS13} on MAX CSPs
and the notion of faithful linear encodings, their corresponding pairs 
of polyhedra, and extended formulations thereof, see~\cite{bfps2012jour}.
Extended formulation for pairs of polyhedra first appeared
in \cite[Section~4.1]{Pashkovich12}
under the name extended relaxation.
The motivation for the streamlined model is to give
a direct natural approach to understanding polyhedral complexity,
in particular in the context of approximate formulations.
Our model is in line with
\cite{CLRS13,LRS15,BPZ2014,BFPS2015}.

In a nutshell, the presented model focuses on the underlying
combinatorial problem, making definitions intuitive,
and avoids technicalities arising from the use of polyhedra,
in particular, the question of the choice of the right polyhedra.
This also eliminates the trap in the polyhedral world of using
facets of the polyhedra not having any combinatorial meaning.
We would like to stress that
there is nothing inherent in the model, which cannot be done
with extended formulation
(with the minimal polyhedral pair, see Remark~\ref{rem:LP-vs-EF}),
i.e., the role of new model is just for streamlining the presentation.

As a particularly relevant example for this paper,
the \emph{stable set polytope} (with a suitable outer polyhedron)
corresponds to the non-uniform version of the \emph{stable set
  problem}.
For the uniform version (see below)
our model naturally provides a formulation,
while there does not seem to be an
immediate natural polyhedral formulation.

As a final remark, the model heavily depends on the original
combinatorial formulation of the problem, e.g., on the choice of
objective functions and feasible solutions,
which is one of the reasons,
why the model is not in conflict with any complexity-theoretic
assumptions; the other one being that
the encoding length of the coefficients are not considered.

We start by defining optimization problems.  For
the sake of exposition, we restrict ourselves to maximization problems.

\begin{defn}[Maximization problem]
  \label{def:max-problem}
  A \emph{maximization problem} $\Pi = (\mathcal S, \mathcal F,{}^*)$ consists of a finite set \(\mathcal{S}\)
  of \emph{feasible solutions}, a finite set \(\mathcal{F}\)
  of \emph{objective functions} where $f \colon \mathcal{S} \to \R$
  for each $f \in \mathcal{F}$, and an \emph{approximation
  guarantee} $f^* \in \R$ so that \(f^* \geqslant \max_{s \in 
  \mathcal{S}} f(s)\) for each objective function \(f \in \mathcal{F}\).
\end{defn}  

By solving such a maximization problem $\Pi = (\mathcal{S},\mathcal{F},{}^*)$ 
we mean determining an approximation \(\widehat{f} \in \R\) of the optimum 
value \(\max_{s \in \mathcal{S}} f(s)\) satisfying
\begin{equation*}
  \max_{s \in \mathcal{S}} f(s) \leqslant \widehat{f} \leqslant f^*,
\end{equation*}
for each \(f \in \mathcal{F}\). Our next definition specifies what it 
means to formulate the maximization problem as a linear program.  We
exemplify the above definition with the maximum
matching problem.

\begin{example}[The maximum weight matching problem]
Let \(G\) be a graph. Then the set of feasible solutions \(\mathcal
S\) is given by all matchings in \(G\) and the set of objective
functions \(\mathcal F\) is given by all weight functions on the
edges. For a matching $s$ and a weight function $f \in \mathcal{F}$,
we define $f(s)$ as the total weight of the edges in the matching.
Finally, the guarantees \(f^*\) are defined as \(f^* \coloneqq
\max_{s \in \mathcal S} f(s)\) for \(f \in \mathcal F\). We obtain 
the \emph{exact} maximum weight matching problem, since the guarantees 
are chosen to be equal to the optimum value. 
\end{example}

We will next specify what it means for a linear program to solve a
maximization problem.

\begin{defn}[LP formulation of a maximization problem]
  \label{def:LP-formulation}
  A \emph{linear programming formulation} of a maximization problem \(\Pi = 
  (\mathcal{S},\mathcal{F},{}^*)\) is a linear system 
  \(\tilde A x \leqslant \tilde b,\ \bar A x =  \bar b\) 
  with \(x \in \R^{d}\) together with \emph{realizations} of:
  \begin{enumerate}
  \item \textbf{feasible solutions} as points \(x^{s} \in \R^{d}\)
    for each \(s \in \mathcal{S}\), such that 
    \begin{equation}
      \label{eq:LP-contain}
      \tilde A x^{s} \leqslant \tilde b,\ \bar A x^{s} =  \bar b;    
    \end{equation}
  \item \textbf{objective functions} as linear functions on $\R^d$, i.e., for
    each \(f \in \mathcal{F}\) there is \(w^{f} \in \R^{d}\) with
    \begin{equation}
      \label{eq:LP-linear}
      \inp{w^{f}}{x^{s}} = f(s) \qquad \text{for all } s \in \mathcal{S}.
    \end{equation}
  \end{enumerate}
  Conditions \eqref{eq:LP-contain} and \eqref{eq:LP-linear} imply that
  \begin{equation}
     \max_{s \in \mathcal{S}} f(s) \leqslant \max \set{\inp{w^{f}}{x}}{\tilde A x \leqslant \tilde b,\ \bar A x =  \bar b} \qquad \text{for all } f \in \mathcal{F}.
  \end{equation} 
  We additionally require that 
  \begin{equation} 
    \label{eq:LP-approx}
    \max \set{\inp{w^{f}}{x}}{\tilde A x \leqslant \tilde b,\
      \bar A x =  \bar b} \leqslant f^*
    \qquad \text{for all } f \in \mathcal{F},  
  \end{equation}
  so that the optimum value of the LP provides an approximation of 
  the optimum value $\max_{s \in \mathcal{S}} f(s)$. The \emph{size} 
  of the formulation is number of inequalities in the LP,
  that is, the number of constraints in \(\tilde A x \leqslant \tilde b\). 
  We define the \emph{formulation complexity} \(\xc(\Pi)\) of the problem \(\Pi\)
  as the minimum size of all its LP formulations.
\end{defn}

\begin{rem}[Relation to extended formulations] 
  \label{rem:LP-vs-EF}
  LP formulations are essentially extended formulations
  of the minimal polyhedral encoding of the problem,
  formulated in a language without the overhead of
  polyhedral concepts.
  Luckily, in most cases actually the minimal polyhedral encoding is
  used, so the results in the two kinds of formulation are easily
  transferable.

  First, let us turn extended formulations into LP formulations.
  If instead of the
constraints $\tilde A x \leqslant \tilde b,\ \bar Ax = \bar b$ we 
used an extended formulation $\tilde E x + \tilde F y \leqslant 
\tilde g,\ \bar E x + \bar F y = \bar g$ to express the set over 
which we want to maximize $\inp{w^f}{x}$, where $y \in \R^k$ is an extra 
(vector) variable, we could redefine the ambient space, the points $x^s$ 
($s \in \mathcal{S}$) and the coefficient vectors $w^f$
($f \in \mathcal{F}$)
to eliminate the extra variable $y$. Indeed, we could work directly in 
$\R^{d+k}$ instead of $\R^d$, consider the vectors $(w^f,0)$ instead of 
$w^f$ and the points $(x^s,y^s)$ instead of $x^s$, where $y^s \in \R^k$ 
is chosen so that $(x^s,y^s)$ satisfies the constraints of the extended
formulation. (Such a point $y^s$ exists because by Condition~\eqref{eq:LP-contain}, 
$x^s$ is feasible for the LP.)

In \cite{bfps2012jour}, we start from a \emph{specific} encoding
of the problem by points $x^s$ ($s \in \mathcal{S}$) and coefficient vectors 
$w^f$ ($f \in \mathcal{F}$), then infer from this a pair $(P,Q)$ of nested 
polyhedra where
\begin{equation}
  \label{eq:minimal-pair}
  \begin{aligned}
    P &\coloneqq \conv{x^{s}}[s \in \mathcal{S}], \\
    Q &\coloneqq \set{x \in \R^d}{\inp{w^{f}}{x}
       \leqslant f^{*},\ \forall f \in \mathcal{F}},
  \end{aligned}
\end{equation}
and finally consider any extended formulation of the pair $(P,Q)$. 
There is a minimal linear encoding, over which
any LP formulation arises as an extended formulation:
let \(\R^{\mathcal{F}}\) be the ambient space, and \(w^{f}_{\min}\)
be the projection to the \(f\)-coordinate for \(f \in \mathcal{F}\).
For every \(s \in \mathcal{S}\),
let \(x_{\min}^{s} \in \R^{\mathcal{F}}\) be the vector with
\(f\)-coordinate \(f(x_{\min}^{s})\) for every \(f \in \mathcal{F}\).
Define the polyhedral pair \((P_{\min}, Q_{\min})\)
by \eqref{eq:minimal-pair}, and this is the minimal encoding,
minimizing the dimension of the problem.
(One might further restrict to the linear space spanned by
\(P_{\min}\), essentially factoring out uninteresting direction in
\(\R^{\mathcal{F}}\),
using linear dependence between objective functions.)

For any LP formulation of the problem given by
linear system
\(\tilde{A} x \leqslant \tilde{b}\), \(\bar{A} x =  \bar{b}\)
with \(x \in \R^{d}\),
feasible solutions \(x^{s}\), and
objective functions \(w^{f}\),
the linear map \(\pi \colon \R^{d} \to \R^{\mathcal{F}}\)
with \(\pi(x)_{f} \coloneqq \inp{w^{f}}{x}\)
and the polyhedron
\(K \coloneqq \{ x \in \R^{d} :
\tilde{A} x \leqslant \tilde{b}\), \(\bar{A} x =  \bar{b}\}\)
defined by the linear system
provides an extended formulation of the same size.
\end{rem}

\subsection{Size lower bounds from nonnegative rank}

The basis of most lower bounds for extended formulations
is Yannakakis's celebrated Factorization Theorem
(see \cite{Yannakakis91,Yannakakis88}),
equating the minimal size of an extended formulation with the 
nonnegative rank of a slack matrix. We will now derive a 
factorization theorem in a similar spirit that characterizes 
formulation complexity, however without requiring an initial
polyhedral representation but rather directly operating on the slack
matrix of a problem. 

A \emph{rank-\(r\)
  nonnegative factorization of \(M \in \R^{m \times n}_+\)}
is a factorization of \(M = TU\)
where \(T \in \R_+^{m \times r}\)
and \(U \in \R_+^{r \times n} \).
This is equivalent to \(M = \sum_{i \in [r]} u_iv_i^\intercal\)
for some (column vectors) \(u_i \in \R_+^{m},\, v_i \in \R_+^{n}\)
with \(i \in [r]\).
The \emph{nonnegative rank of \(M\)},
denoted by \(\nnegrk{M}\),
is the minimum \(r\)
such that there exists a rank-\(r\)
nonnegative factorization of \(M\).

We will use the following elementary properties of nonnegative rank.

\begin{lem}
  \label{lem:nnegrk-preserve}
  Let \(A \in \R_{+}^{k \times m}\), \(M \in \R_{+}^{m \times n}\)
  and \(B \in \R_{+}^{n \times l}\) be nonnegative matrices.
  Then \(\nnegrk (AMB) \leqslant \nnegrk M\).
  In particular, the following operations do not increase
  the nonnegative rank of a matrix:
  \begin{enumerate}
  \item deleting, duplicating, permuting rows or columns;
  \item adding a nonnegative linear combination of rows
    as a new row;
  \item adding a nonnegative linear combination of columns
    as a new column.
  \end{enumerate}
\end{lem}

By considering the nonnegative factorizations $M = I_mM$ and $M = MI_n$, 
we immediately obtain 
\begin{equation}
  \label{eq:rkFromVertex}
\nnegrk M \leqslant \min\{m,n\}.
\end{equation}

In analogy to extended formulations, central for understanding the
size of a linear program will be the concept of the
slack matrix:

\begin{defn}[Slack matrix]
 Let \(\Pi = (\mathcal{S},\mathcal{F},{}^*)\) be a maximization problem as 
  in Definition~\ref{def:max-problem}. The \emph{slack matrix} of \(\Pi\) 
  is the nonnegative \(\mathcal{F} \times \mathcal{S}\) matrix \(M\),
  with entries
  \begin{equation*}
    M(f, s) \coloneqq f^{*} - f(s).
  \end{equation*}  
\end{defn}

We are ready to formulate the factorization theorem for formulation
complexity. 
We provide a proof for the sake of completeness,
even though in the equivalent formulation for polyhedral pairs,
it already appeared in \cite{Pashkovich12}.

\begin{thm}[Factorization theorem for formulation complexity]
  (C.f., \cite[Lemmas~4.1 and ~4.2]{Pashkovich12}.)
  \label{thm:factorization}
  Consider a maximization problem \(\Pi = (\mathcal{S},\mathcal{F},{}^*)\) as 
  in Definition~\ref{def:max-problem} with slack matrix \(M\).
  Then
  \begin{equation*}
    \nnegrk M - 1 \leqslant \xc(\Pi) \leqslant \nnegrk M.
  \end{equation*}
\begin{proof}
To prove the first inequality, let 
$\tilde A x \leqslant \tilde b,\ \bar A x = \bar b$
be an arbitrary size-\(r\) LP formulation of \(\Pi\),
with realizations \(w^{f}\) (\(f \in \mathcal{F}\))
of objective functions and \(x^{s}\) (\(s \in \mathcal{S}\))
of feasible solutions.
We shall construct a size-\((r+1)\)
nonnegative factorization of \(M\).
As \(\max_{x : Ax \leqslant b} \inp{w^{f}}{x} \leqslant f^{*}\)
by Condition~\eqref{eq:LP-approx},
via Farkas's lemma, we have
\[
f^{*} - \inp{w^{f}}{x} = T(f, 0) + \sum_{j=1}^r T(f,j) \left( \tilde
  b_j - \tilde A_{j} x \right) + \sum_{j=1}^{k} \lambda(f,j) (
\bar b_{j} - \bar A_{j} x)
\]
for some nonnegative multipliers $T(f,j) \in \R_+$ with \(0 \leqslant j
\leqslant r\),
and arbitrary multipliers \(\lambda(f,j)\) for \(0 \leqslant j \leqslant k\).
By taking $x = x^s$, and by \(\bar A x^{s} = \bar b\),
we obtain
\begin{align}
  \label{eq:1}
  M(f,s) &= \sum_{j=0}^r T(f,j) U(j,s),
  &
  \text{with} \quad
  U(j,s) &\coloneqq
  \begin{cases}
    1 & \text{for } j = 0,
    \\
    \tilde b_j - \tilde A_{j} x^{s} & \text{for } j > 0.
  \end{cases}
\end{align}
That is, $M = TU$. By construction, $T$ is nonnegative.
By Condition~\eqref{eq:LP-contain} we also obtain that $U$ is nonnegative.
Therefore $M = TU$ is a rank-$(r+1)$ nonnegative factorization of $M$.

For the second inequality,
let \(M = TU\) be a size-\(r\) nonnegative factorization.
We shall construct an LP formulation of size \(r\).
Let \(T^{f}\) denote the \(f\)-row of \(T\) for \(f \in \mathcal{F}\),
and \(U_{s}\) denote
the \(s\)-column of \(U\) for \(s \in \mathcal{S}\).
Then
\begin{equation}
  \label{eq:factor-nonneg}
  T^{f} U_{s} = M(f, s) = f^{*} - f(s).
\end{equation}

In the following we represent the vectors \(y \in \R^{r+1}\) via \(y =
(x, \alpha)\) with \(x \in \R^{r}\) and \(\alpha \in \R\).
We claim that the linear system
\begin{equation}
  \label{eq:factor-LP}
  x \leqslant 0,\ \alpha = 1
\end{equation}
with representations
\begin{align*}
  w^{f} &\coloneqq (T^{f}, f^{*}) \qquad \forall f \in \mathcal{F}
  &&\text{and}&
  x^{s} &\coloneqq ( -U_{s}, 1) \qquad \forall s \in \mathcal{S}
\end{align*}
satisfies the requirements of
Definition~\ref{def:LP-formulation}. 
Condition~\eqref{eq:LP-linear} clearly follows
from \eqref{eq:factor-nonneg}:
\begin{equation*}
  \inp{w^{f}}{x^{s}} = -T^{f} U_{s} + f^{*} = f(s).
\end{equation*}
Moreover, the \(x^{s}\) satisfy the linear program
\eqref{eq:factor-LP},
because \(U\) is
nonnegative, so that Condition~\eqref{eq:LP-contain} is fulfilled.
Finally, Condition~\eqref{eq:LP-approx} also follows readily:
\begin{equation*}
  \max \set{\inp{w^{f}}{(x, \alpha)}}{x \leqslant 0, \alpha =1}
  = \max \set{T^{f} x + f^{*}}{x \leqslant 0} = f^{*},
\end{equation*}
as the nonnegativity of \(T\) implies \(T^{f} x \leqslant 0\); equality
holds e.g., for \(x = 0\).
Thus we have constructed an LP formulation with \(r\) inequalities,
as claimed.
\end{proof}
\end{thm}

\subsection{Maximum stable set problems}
\label{sec:maximum-stable-set}
Now we describe two ways in which the maximum stable set problem can be
seen as a maximization problem $\Pi = (\mathcal{S},\mathcal{F},{}^*)$ that
we use in the rest of the paper.

  \begin{defn}[The maximum stable set problem—uniform model]
    \label{def:STABu}
    We start with a family of graphs \(\mathcal{G}\) with \(V(G) 
    \subseteq [n]\) for each $G \in \mathcal{G}$. The goal is to 
    approximate the stability number of the graphs in $\mathcal{G}$ 
    within a given relative error guarantee \(\rho \geqslant 0\). 
    To each graph $G \in \mathcal{G}$ we correspond an objective 
    function $f_G$ in the set $\mathcal{F}$ of objective functions. 
    The set \(\mathcal{S}\) of feasible solutions is taken to be 
    the set of \emph{all} subsets of \([n]\). This is natural since, 
    typically, $\mathcal{G}$ contains many graphs with many
    different vertex sets and many $S \subseteq [n]$ occur
    as a stable set in some \(G \in \mathcal{G}\). We require
    \begin{enumerate}[label=(\alph*)]
    \item \label{cond:effectiveness}
      $\alpha(G) = \max_{S \in \mathcal{S}} f_{G}(S)$,
      for all $G \in \mathcal{G}$;
    \item \label{cond:interpretability}
      $f_{G}(S) = \size{V(G) \cap S}$ 
      whenever $\size{V(G) \cap S} \leqslant 1$,
      for all $G \in \mathcal{G}$ and $S \in \mathcal{S}$.
    \end{enumerate}
    We choose as approximation guarantee \(f_{G}^{*} \coloneqq (1 +
    \rho) \alpha(G)\) for \(G \in \mathcal{G}\). This defines a
    \emph{class} of maximization problems that we denote by
    \(\STABu(\mathcal{G}, \rho)\): every set \(\mathcal F\)
    of objective functions subject to the above conditions gives 
    a valid approximate computation of the stability number in a 
    graph $G$ chosen from $\mathcal{G}$. The
    later derived lower bounds apply to every problem
    \(\Pi\) from this class and we define $\xc(\STABu(\mathcal{G}, \rho)) 
    \coloneqq \min \set{\xc(\Pi)}{\Pi \in \STABu(\mathcal{G}, \rho)}$.
    In the exact case, that is, when $\rho = 0$, we use 
    $\STABu(\mathcal{G})$ to mean $\STABu(\mathcal{G},0)$. 
\end{defn}
    
As an example, a concrete maximization problem $\Pi \in \STABu(\mathcal{G}, \rho)$
is obtained by defining the objective functions as
\begin{equation}
  \label{eq:obj-STAB}
  f_{G}(S) \coloneqq \size{V(G) \cap S} - \size{E(G[S])}.
\end{equation}
Here \(f_{G}(S)\) is a conservative lower bound on the size 
of a stable set, which one would naturally obtain from \(G[S]\)
by deleting one endpoint from every edge in that induced subgraph.
Moreover we have \(f_{G}(S) = \size{S}\) for stable sets \(S\) of
\(G\), i.e., in this case our choice is exact.  We leave it to the 
reader to check Condition \ref{cond:effectiveness} while Condition~\ref{cond:interpretability} is immediate.    

\begin{rem}[Correlation polytope as a specific uniform encoding]
  \label{rem:uniform-STAB-EF}
  In the uniform model, the feasible solutions are
  all vertex sets to include every possible stable sets.
  Therefore, taking into account the form of objective functions,
  in the polytope world the uniform model
  corresponds to a polyhedral pair
  where the inner polytope is the correlation polytope.
  Recall that
  the stable set polytope deals with stable sets of a fixed graph,
  and therefore corresponds to the non-uniform model defined below.
\end{rem}
    
  \begin{defn}[The maximum stable set problem—non-uniform model]
    \label{def:STAB(G)}
    To each \emph{fixed} graph $G$ on $n$ vertices we associate a
    problem that corresponds to the exact computation of the stability
    number of some induced subgraph $H$ of $G$. We will denote this 
    problem by \(\STABnu(G)\) or (later) simply $\STAB(G)$. More 
    precisely, $\STABnu(G)$ has feasible solutions that are sets 
    $S \subseteq V(G)$ which are stable sets of \(G\), and objective
    functions of the form $f_H(S) \coloneqq \size{S \cap V(H)}$
    where $H$ is an induced subgraph of $G$.
    We let $f_H^* \coloneqq \alpha(H)$ for the
    approximation guarantee.
  \end{defn}

\begin{rem}[Stable set polytope as a specific non-uniform encoding]
Recall that the stable set polytope of \(G\) is
the convex hull of (a \emph{specific encoding} of) all stable sets of \(G\),
i.e., the feasible solutions of the non-uniform problem.
The objective functions are linear over the stable set polytope,
but need not include all facet-defining linear functions,
hence the non-uniform problem corresponds to a polyhedral pair
where the inner polytope is the stable set polytope.

In particular,
the abuse of notation by using $\STAB(G)$ to denote both
the stable set polytope of \(G\) and
the problem $\STABnu(G)$ is not too severe.
\end{rem}

\subsection{Unique disjointness}
\label{sec:unique-disjointness}

As we will demonstrate, the polyhedral hardness of the maximum stable set 
problem arises from the unique disjointness (partial) matrix. Recall 
that the unique disjointness (\UDISJ) matrix,
which we denote by \(\UDISJ[n]\)
below, has \(2^{n}\) rows and \(2^{n}\) columns indexed by 
all size-$n$ \(0/1\)-vectors \(a\) and \(b\).
Its entries are:
\begin{equation}
\label{eq:UDISJ}
  \UDISJ[n](a, b) = \begin{cases}
0 &\text{if } a^\intercal b = 1\\ 
1 &\text{if } a^\intercal b = 0.
\end{cases}
\end{equation}
Although \(\UDISJ[n]\) is only a partial matrix,
i.e., not all of its entries
are defined, we will refer to it as a matrix from here on. The fact
that it is only partial does not matter for our purpose, as we only 
care for whether this (partial) matrix occurs as a submatrix of 
some appropriate slack matrices. The \UDISJ matrix has been studied 
in many disciplines, arguably the most notable being communication 
complexity. 

\begin{thm} \label{thm:corLB} 
\(\nnegrk{\UDISJ[n]} \geqslant 2^{n \cdot\log (3/2)}.\)
\end{thm}

The factor \(\log (3/2) \approx 0.585\) in the exponent is the
current best one due to \cite{KW14}; for various 
approximate case versions see \cite{bfps2012jour,braverman2012information,%
BP2013}. The first exponential lower bound was established in \cite{extform4,
FMPTW15jour} by combining the seminal work of \cite{Razborov92} together with an 
observation in \cite{Wolf03}. This was at the core of the first 
results establishing high extension complexity for the correlation polytope, 
cut polytope, stable set polytope, and the TSP polytope in~\cite{extform4,FMPTW15jour}. 
  
\cite{bfps2012jour} prove that \emph{any} $2^n \times 2^n$ 
matrix $M$ with rows and columns indexed by vectors in $\{0,1\}^n$
satisfying \eqref{eq:UDISJ} has superpolynomial nonnegative rank, and 
that this remains true even if we shift the entries of the matrix $M$ 
by some number $\rho = O(n^{1/2-\varepsilon})$. This result was then
extended to shifts \(\rho = O(n^{1-\varepsilon})\) in 
\cite{braverman2012information} which then immediately leads
to a polyhedral inapproximability of CLIQUE (in the uniform model!) 
of \(O(n^{1-\varepsilon})\), matching Håstad's hardness result 
for approximating CLIQUE.

The \emph{$\rho$-shifted \UDISJ matrix} is any $2^n \times 2^n$ matrix 
indexed by pairs $(a,b)$ where $a, b \in \{0,1\}^n$ such that 
\begin{equation}
(\UDISJ(n)+\rho J)(a,b) = \begin{cases} 
\rho &\text{if } a^\intercal b = 1\\ 
1+\rho &\text{if } a^\intercal b = 0,
\end{cases}
\end{equation}
where $J$ is the all-one matrix of compatible size.

In \cite{BP2013} an information-theoretic approach for studying the
nonnegative rank has been developed. This approach allows to lower 
bound the nonnegative rank of various \lq{}deformations\rq{} of the 
\UDISJ matrix. The following theorem from \cite{BP2013} will allow
us to analyze a specific type of deformation that we will use in 
the following.  Informally speaking, the theorem shows that the \UDISJ
matrix has high nonnegative rank almost everywhere. Below, we use
$\UDISJ(n,k)$ to denote the \UDISJ matrix $\UDISJ(n)$ restricted to 
subsets of size $k$.

\begin{thm}
\label{thm:UDISJ_fraction_of_rows} 
  Let $M$ be any submatrix of the $\rho$-shifted \UDISJ
  matrix $\UDISJ(n,k) + \rho J$ obtained by deleting at 
  most an $\alpha$-fraction of rows and at most a 
  $\beta$-fraction of columns for some $0 \leqslant \alpha, 
  \beta < 1$. Then for $0 < \varepsilon < 1$:
\[
\nnegrk M \geqslant 
2^{(1/8 (\rho + 1) - (\alpha+\beta)\bentropy{1/4})n
  - O(n^{1-\varepsilon})}
\quad \text{for} \quad k = n/4 + O(n^{1-\varepsilon}).
\]
\end{thm}
Here \(\bentropy{\cdot}\) is the binary entropy function,
in particular \(\bentropy{1/4} \approx 0.811\).

We finish this section with an easy example of embedding $\UDISJ(n)$
into a slack matrix.

\begin{example} 
  Let us consider the family $\mathcal{G}$ of all (non-empty) 
  complete graphs, that is, all graphs $G$ with $V(G) \subseteq [n]$ 
  and $\alpha(G) = 1$ (in case $V(G) = \emptyset$ we have $\alpha(G) = 0$). 
  Then the slack matrix $M$ of any \(\Pi \in \STABu(\mathcal{G})\) 
  contains $\UDISJ(n)$ as a submatrix (without the row of the empty 
  set $a = \emptyset$). To verify this, note that by 
  Condition~\ref{cond:interpretability} for every pair of
  subsets \(a, b \subseteq [n]\) (with \(a \neq \emptyset\)):
  \[
  M(K_{n}[a], b) = 1 - \size{a \cap b} =
  \begin{cases}
    0 &\text{if } \size{a \cap b} = 1\\
    1 &\text{if } \size{a \cap b} = 0,
  \end{cases}
  \]
  where $K_{n}$ denotes the complete graph on $[n]$.
\end{example}

\section{Average case complexity in the uniform model}
\label{sec:aver-case-compl-unif}

We will now establish our main result regarding the average case
complexity of the uniform model. We obtain that for any random collection 
of graphs $\mathcal{G} = \mathcal{G}(n,p)$ where each graph $G$ with 
$V(G) \subseteq [n]$ is picked independently with probability \(p\), the 
formulation complexity of $\STABu(\mathcal{G})$ is high. Loosely 
speaking, the size of any ``simultaneous'' LP formulation of the 
maximum stable set problem for all graphs in this random collection 
$\mathcal{G}$ is high. In a way, this indicates that the instances of 
the stable set problem resulting in high polyhedral complexity are not 
localized in a set of small density.

\begin{mainthm}[Super-polynomial fc of $\STABu(\mathcal{G})$ w.h.p.]
\label{mainthm:uniformModel}
  Let $n \geqslant 40$ and $p \in [0,1]$ with
  $p \geqslant  2^{- \binom{n/4}{2} + n}$.
  Pick a random family \(\mathcal{G} = \mathcal{G}(n,p)\) of 
  graphs by adding each graph $G$ with $V(G) \subseteq [n]$
  to the family with probability $p$, independent of the other 
  \(G\). Then
  \begin{equation*}
    \probability{\xc(\STABu(\mathcal{G})) \geqslant 2^{\Omega(n/\log n )}}
    \geqslant
    1 - 2^{-2^n}.
  \end{equation*}
\end{mainthm}

A crucial point of the proof is
a concentration result on \(\alpha(G)\).
It is well-known that
almost all graphs $G$ on $n$ vertices have stability
number $\alpha(G) \sim 2 \log n$.
However, the following rough estimate will be sufficient for
our purpose, see e.g.\ \cite[Proposition 11.3.4, page 304]{Diestel}
for a proof.

\begin{lem}
\label{lem:rough}
Let $n \geqslant 10$. The probability that a uniformly sampled
random graph $G$ with $V(G) = [n]$ has $\alpha(G) \geqslant 3 \log n$ 
is at most $n^{-1}$.
\end{lem}

We are ready to prove the main theorem of this section.

\begin{proof}[Proof of Main Theorem~\ref{mainthm:uniformModel}]
Consider any problem \(\Pi \in \STABu(\mathcal{G})\). Thus, 
$\Pi = (\mathcal{S},\mathcal{F},{}^*)$ is a maximization problem 
that corresponds to determining the stability number of the graphs
of $\mathcal{G}$. Recall that $\mathcal{S}$ is the set of \emph{all} 
subsets of $S \subseteq [n]$ and that $\mathcal{F}$ contains an objective 
function $f_G$ for each graph $G \in \mathcal{G}$. These functions
satisfy $\max_{S \in \mathcal{S}} f_G(S) = \alpha(G)$ 
(see Condition~\ref{cond:effectiveness}). Also, recall that we 
require $f_{G}(S) = \size{V(G) \cap S}$ whenever $\size{V(G) 
\cap S} \leqslant 1$ (see Condition~\ref{cond:interpretability}).
Finally, since $\rho = 0$, the approximation guarantee is 
$f^*_G = \alpha(G)$.

Consider the slack matrix \(M\) of $\Pi$. We want to show that 
the nonnegative rank of \(M\) is high by embedding a large portion 
of \UDISJ(n) into it. By Theorem~\ref{thm:factorization}, this will 
imply that $\xc(\Pi)$ and thus $\xc(\STABu(\mathcal{G}))$ is high, since 
$\Pi \in \STABu(\mathcal{G})$ is arbitrary.

The main idea of the proof is that, with extremely large probability,
among all sets $a \subseteq [n]$
of size $\roundup{n/4}$, the collection $\mathcal{G}$ contains many
graphs $G$ with $V(G) = a$. For each of these sets $a$, there will 
be at least one corresponding graph $G_{a}$
with
$\alpha(G_{a}) \leqslant 3 \log n$.
Restricting to these graphs \(G_{a}\),
the resulting slack matrix contains
a large part of the $O(\log n)$-shifted
\UDISJ matrix as a submatrix (in fact, a large fraction of the rows,
and all the columns, survive). 
We apply Theorem~\ref{thm:UDISJ_fraction_of_rows} to conclude.

Now let us turn to the detailed proof.
Consider a set $a$ with $a \subseteq [n]$ and size $k \coloneqq 
\roundup{n/4}$. We say that a graph $G$ is \emph{good} for $a$ 
if $V(G) = a$ and $\alpha(G) \leqslant 3 \log n$. Set $a$ is 
said to be \emph{good} if some graph $G \in \mathcal{G}$ is 
good for $a$. Otherwise, $a$ is called \emph{bad}. 

We claim that, with high probability, the total fraction of bad
sets among all $k$-sets $a$ is at most $\alpha \coloneqq 1/(24 \log n)$.
By Lemma \ref{lem:rough}, the total number of graphs $G$ with $V(G) =
a$ that are not good for a fixed $k$-set $a$ is at most $k^{-1}
2^{\binom{k}{2}}$. Thus
\begin{equation}
  \label{eq:bad-graphs}
 \begin{split}
\probability{a \text{ is bad}} &= 
\probability{\mathcal{G} \text{ contains no good graph for } a}\\
&\leqslant (1-p)^{(1 - k^{-1}) 2^{\binom{k}{2}}}\\ 
&\leqslant \mathrm{e}^{-p (1 - k^{-1}) 2^{\binom{k}{2}}}\\
&\leqslant 2^{-\frac{9}{10} 2^{n} \log \mathrm{e} }\\
&\leqslant \alpha \, 2^{-2^n}.
 \end{split}
\end{equation}
where the third inequality follows from $k \geqslant n/4 \geqslant 10$ 
and $p \geqslant 2^{-\binom{n/4}{2}+n}$ and the last inequality 
follows from our choice of $\alpha$. Let $X$ denote the random 
variable with value the number of bad $k$-sets $a$. By Markov's 
inequality, 
\[
\probability{X \geqslant \alpha \binom{n}{k}} \leqslant 2^{-2^n}.
\]

For each good $k$-set $a$, pick a good graph $G_{a} \in \mathcal{G}$
arbitrarily, i.e., with \(\alpha(G_{a}) \leqslant 3 \log n\).
We need two auxiliary nonnegative vectors.
Let \(\mathbf{1}\) denote the row vector with all 
entries \(1\), and entries indexed by subsets \(S \subseteq [n]\).
Let \(u\) be the vector with entries indexed by the
\(G_{a}\) for all good \(k\)-sets \(a\), and with entries
\begin{equation*}
  u_{G_{a}} \coloneqq 3 \log n - \alpha(G_{a}) \geqslant 0.
\end{equation*}
Because the slack matrix $M$ of $\Pi$ satisfies
\[
M(G_{a},S) =
\begin{cases}
  \alpha(G_{a}) - 1 & \text{if } \size{V(G) \cap S} = 1, \\
  \alpha(G_{a})     & \text{if } \size{V(G) \cap S} = 0,
\end{cases}
\]
we obtain after applying the rank-1 shift \(u\mathbf{1}\)
\begin{equation*}
  (M + u \mathbf{1})(G_{a}, S) =
\begin{cases}
  3 \log n - 1 & \text{if } \size{V(G) \cap S} = 1, \\
  3 \log n     & \text{if } \size{V(G) \cap S} = 0.
\end{cases}
\end{equation*}
By the above, with probability at least $1 - 2^{-2^n}$,
the fraction of bad $k$-sets $a \subseteq [n]$ among all \(k\)-sets
is at most $\alpha$, and hence the matrix \((M + u \mathbf{1})\)
contains a $(3 \log n - 1)$-shift of \(\UDISJ(n,k)\), with at most 
an $\alpha$-fraction of the rows thrown away. From Theorem
\ref{thm:UDISJ_fraction_of_rows} (with \(\beta =0\)), the nonnegative 
rank of \(M\) is at least
\[
\nnegrk M \geqslant \nnegrk (M + u \mathbf{1}) - 1
\geqslant
2^{(1/8 (3 \log n + 1) - \alpha \bentropy{1/4}) \cdot n
  - O(n^{1-\varepsilon})} - 1
  = 2^{\Omega(n / \log n)}.
\]
\end{proof}

Without much additional work, we can obtain a similar lower bound 
on the average case formulation complexity also in the approximate
case, that is, when $\rho > 0$.

\begin{cor}[Super-polynomial xc of $\STABu(\mathcal{G},\rho)$ w.h.p.]
\label{cor:averCompApproxUniform}
  As in Main Theorem~\ref{mainthm:uniformModel},
  let \(\mathcal{G} = \mathcal{G}(n,p)\) be a random family of graphs
  such that each graph \(G\) with \(V(G) \subseteq [n]\) 
  is contained in \(\mathcal{G}\) with probability \(p \geqslant  2^{- \binom{n/4}{2} + n}\) 
  independent of the other graphs.
  Then for all \(0 < \varepsilon < 1/2\) and
  \(\rho \leqslant  \frac{n^{1-\varepsilon}}{\log
      n}\), we have that \(\STABu(\mathcal{G}, \rho)\) 
  has formulation complexity
$2^{\Omega(n^{\varepsilon})}$, with probability at least $1 -
2^{-2^n}$.
\begin{proof}
The proof is identical to Theorem~\ref{mainthm:uniformModel}
subject to minor changes.
First, now we use a different factor
\(\alpha \coloneqq 1 / [24 (1 + \rho) \log n]\).
The computation in Equation~\eqref{eq:bad-graphs}
still remains valid, as the different choice of \(\alpha\)
affects only the last inequality,
which remains true,
because \(\rho \leqslant n^{1 - \varepsilon } / \log n\).
Therefore again
with probability \(1 - 2^{- 2^{n}}\),
with the exception of at most
an \(1/[24 (1 + \rho) \log n]\)-fraction,
all graphs \(G \in \mathcal{G}\) are good,
i.e., have clique number at most
\(\alpha(G) \leqslant 3 \log n\).

The second difference is that due to dilation,
the slack entries of \(M\) are a bit different:
\begin{equation*}
  M(G,S)
  =
  \begin{cases}
    (1 + \rho) \alpha(G) - 1 & \text{if } \size{V(G) \cap S} = 1, \\
    (1 + \rho) \alpha(G)     & \text{if } \size{V(G) \cap S} = 0.
  \end{cases}
\end{equation*}
This provides an embedded copy of
a \((3 (1 + \rho) \log n - 1)\)-shift of $\UDISJ(n,k)$
in \(M + (1 + \rho) u \mathbf{1}\)
with at most an \(\alpha\)-fraction of rows missing.
Hence Theorem~\ref{thm:UDISJ_fraction_of_rows} applies again,
but now we replace the \(\varepsilon\) there
with a \(\varepsilon'\)
lying strictly between \(\varepsilon\) and \(1/2\),
but not depending on \(n\).
(E.g., \(\varepsilon' \coloneqq (\varepsilon + 1/2) / 2\) is
a good choice.)
This will make the error term
\(O(n^{1 - \varepsilon'}) = o(n^{1 - \varepsilon})\)
in the exponent negligible.
We obtain the lower bound:
\(
2^{(1/8 (3 (1 + \rho) \log n + 1) - \alpha \bentropy{1/4}) \cdot n
  - O(n^{1 -\varepsilon'})} - 1
  = 2^{\Omega(n^{1 - \varepsilon})}
\)
on formulation complexity,
as claimed.
\end{proof}
\end{cor}

Observe that the relative approximation guarantee $\rho$ in
 Corollary~\ref{cor:averCompApproxUniform} can be larger than \(3 \log
 n\). The reason why this is possible, contradicting initial intuition,
 is that the hardness arises from having
 many different graphs and hence many objective functions to
 consider simultaneously and the encoding is highly non-monotone.
 Roughly speaking, graphs with different vertex sets are independent 
 of each other, even if one is an induced subgraph of the other.

\section{Average case complexity in the non-uniform model}
\label{sec:aver-case-compl-nonunif}

We now turn our attention to the non-uniform problem $\STABnu(G)$,
see Definition~\ref{def:STAB(G)}. Thus $G$ is a fixed graph with 
vertex set $[n]$, the feasible solutions are the stable sets of $G$
and the objective functions correspond to induced subgraphs of $G$.
For simplicity of notation, we index the objective functions with 
the supporting vertex set \(a \subseteq [n]\) instead of the induced 
subgraph \(G[a]\). Thus $f_{a}(S) = \size{S \cap a}$ for every 
stable set $S$ of $G$ and $a \subseteq [n]$, and $f^*_{a} = 
\alpha(G[a])$. 

Notice that an LP formulation for the problem $\STABnu(G)$ is provided 
by the linear description of the stable set polytope of \(G\), or any
extended formulation of the stable polytope of \(G\). In this 
sense, $\STABnu(G)$ generalizes the stable set polytope. For 
the sake of brevity, we denote the problem $\STABnu(G)$ 
simply by \(\STAB(G)\).

We lower bound the formulation complexity of \(\STAB(G(n,p))\)
for the \emph{random} Erdős–Rényi graph \(G(n,p)\).
Our strategy is to embed certain subdivisions of the 
complete graph $K_t$ as induced subgraphs of $G$, with $t$ 
as large as possible, using the probabilistic method.

Our construction is parametrized by an \emph{even} integer 
$\ell \geqslant 0$. For a graph \(T\), we let \(\gadget{T}\) denote 
the subdivision of \(T\) obtained by replacing each edge \(ij\) of 
\(T\) with a path $P_{ij}$ with $2\ell + 3$ edges between $i$ and $j$. 
We denote $u_{ij}$ and $v_{ij}$ the middle vertices of $P_{ij}$, see 
Figure~\ref{fig:gadget}. In total, \(\gadget{T}\) has \(v \coloneqq 
\card{V(T)} + (2 \ell + 2) \card{E(T)}\) vertices and \(e \coloneqq 
(2 \ell + 3) \card{E(T)}\) edges. 

\begin{figure}[hbt]
  \centering
  \tikzset{
    vertex/.style={circle,draw},
    dots/.style = {
      to path={-- node[midway, rectangle, fill=white,
        label=above:{#1}]{\(\dots\)}
    (\tikztotarget)}}}
  \begin{tikzpicture}
    \node[vertex, label=above:{\(i\)}] (i) {};
    \node[vertex, label=above:{\(u_{ij}\)}] (u_ij)
         [right=of i] {};
    \node[vertex, label=above:{\(v_{ij}\)}] (v_ij)
         [right=of u_ij] {};
    \node[vertex, label=above:{\(j\)}] (j)
         [right=of v_ij] {};
    \draw (i) -- (u_ij) -- (v_ij) -- (j);
  \end{tikzpicture}
  \caption{Path \(P_{ij}\) replacing edge \(ij\) of \(T\)
    in $\gadget{T}$. There are $\ell + 1$ edges between 
    \(i\) and \(u_{ij}\), as between \(v_{ij}\) and \(j\).
    In the figure, $\ell = 0$.}
\label{fig:gadget}
\end{figure}
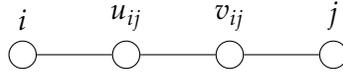

Our next lemma proves that increasing the parameter $\ell$ decreases 
the average degree of induced subgraphs of the gadget graph $\gadget{T}$, 
which makes it easier to embed $\gadget{T}$ in $G(n,p)$ for lower values 
of $p$. 

\begin{lem}
  \label{lem:avg-deg-gadget}
  For any graph \(T\),
  the average degree of any induced subgraph of \(\gadget{T}\)
  is at most \(2 + 1 / (\ell + 1)\). For \(\ell = 0\), the average 
  degree is at most \(3\).
\end{lem}

\begin{proof}
Consider an induced subgraph of \(\gadget{T}\) with 
maximum average degree. Clearly, we may assume that 
$T$ contains at least one edge, so the average degree
of the induced subgraph is at least $3/2$. This implies
that it does not contain any vertex
of degree at most $1$, because the deletion of any such 
vertex would increase the average degree. It follows
that the induced subgraph is of the form \(\gadget{H}\)
where \(H\) is a subgraph of \(T\). The average degree
of \(\gadget{H}\) can be expressed in terms of that of 
\(H\) as:
\[
\avgd(\gadget{H}) = \frac{2 |E(\gadget{H})|}{|V(\gadget{H})|} = 
\frac{2(2 \ell + 3) |E(H)|}{|V(H)| + (2 \ell + 2) |E(H)|}
= \frac{(2 \ell + 3) \avgd(H)}{1 + (\ell + 1) \avgd(H)}.
\]
From the last expression we see that the average degree of 
\(\gadget{H}\) is an increasing function of the average degree 
of \(H\) that tends to \((2 \ell + 3) / (\ell + 1) = 2 + 1 / 
(\ell + 1)\) in the limit.
\end{proof}

\begin{lem}
  \label{lem:xc-gadget}
  If graph \(G\) contains $\gadget{K_{t}}$
  as an induced subgraph, then
  \[
    \xc(\STAB(G)) \geqslant \nnegrk \UDISJ[t]
    - 2
    \geqslant 2^{t \log (3/2)} - 2.
  \]
\begin{proof}
Let \(T\) be any graph with \(\gadget{T}\)
being an induced subgraph of \(G\).
Later we will specialize to \(T = K_{t}\).
We choose representatives of subsets of \(V(T)\)
as stable sets and induced subgraphs of \(\gadget{T}\).

For every \(b \subseteq V(T)\), we choose an extension 
to a stable set \(S(b)\) of \(\gadget{T}\) by adding
as much internal vertices of each path \(P_{ij}\) as possible, 
see Figure~\ref{fig:extension}. Let $ij \in E(T)$. On the part 
of $P_{ij}$ from $i$ to $u_{ij}$, as well on the part from $j$ 
to $v_{ij}$, we alternate between vertices belonging to $\gadget{T}$ 
and not belonging to $\gadget{T}$, with one exception: if both 
$i, j \in b$ then we drop either $u_{ij}$ or $v_{ij}$ in $S(b)$.

If we start from a maximum stable set \(b\) of \(T\), we see
that \(S(b)\) has \(\size{b} + (\ell + 1) |E(T)|\) vertices.
Thus \(\alpha(\gadget{T}) \geqslant \alpha(T) + (\ell + 1) |E(T)|\).
This inequality is tight because no stable set \(S\) of \(\gadget{T}\)
can have more vertices than \(S(b)\),
where \(b \coloneqq S \cap V(T)\).
That is, \(\size{S} \leqslant \size{b} + (\ell + 1) |E(T)| - |E(T[b])|
\leqslant \alpha(T) + (\ell + 1) |E(T)|\).
For any graph \(T\), we get:
\begin{equation}
\label{eq:alpha-gadget}
\alpha(\gadget{T}) = \alpha(T) + (\ell + 1) |E(T)|.
\end{equation}

\begin{figure}[hbt]
  \centering
  \tikzset{
    vertex/.style={circle,draw},
    dots/.style = {
      to path={-- node[midway, rectangle, fill=white,
        label=above:{#1}]{\(\dots\)}
    (\tikztotarget)}}}
  \begin{tikzpicture}
    \node[vertex, label=above:{\(i\)}] (i) {};
    \node[vertex, fill, label=above:{\(u_{ij}\)}] (u_ij)
         [right=of i] {};
    \node[vertex, label=above:{\(v_{ij}\)}] (v_ij)
         [right=of u_ij] {};
    \node[vertex, label=above:{\(j\)}] (j)
         [right=of v_ij] {};
    \draw (i) -- (u_ij) -- (v_ij) -- (j);
  \end{tikzpicture}
  \quad
  or
  \quad
  \begin{tikzpicture}
    \node[vertex, label=above:{\(i\)}] (i) {};
    \node[vertex, label=above:{\(u_{ij}\)}] (u_ij)
         [right=of i] {};
    \node[vertex, fill, label=above:{\(v_{ij}\)}] (v_ij)
         [right=of u_ij] {};
    \node[vertex, label=above:{\(j\)}] (j)
         [right=of v_ij] {};
    \draw (i) -- (u_ij) -- (v_ij) -- (j);
  \end{tikzpicture}
  \\[2ex]
  \begin{tikzpicture}
    \node[vertex, label=above:{\(i\)}] (i) {};
    \node[vertex, fill, label=above:{\(u_{ij}\)}] (u_ij)
         [right=of i] {};
    \node[vertex, label=above:{\(v_{ij}\)}] (v_ij)
         [right=of u_ij] {};
    \node[vertex, fill, label=above:{\(j\)}] (j)
         [right=of v_ij] {};
    \draw (i) -- (u_ij) -- (v_ij) -- (j);
  \end{tikzpicture}
  \\[2ex]
  \begin{tikzpicture}
    \node[vertex, fill, label=above:{\(i\)}] (i) {};
    \node[vertex, label=above:{\(u_{ij}\)}] (u_ij)
         [right=of i] {};
    \node[vertex, fill, label=above:{\(v_{ij}\)}] (v_ij)
         [right=of u_ij] {};
    \node[vertex, label=above:{\(j\)}] (j)
         [right=of v_ij] {};
    \draw (i) -- (u_ij) -- (v_ij) -- (j);
  \end{tikzpicture}
  \\[2ex]
  \begin{tikzpicture}
    \node[vertex, fill, label=above:{\(i\)}] (i) {};
    \node[vertex, label=above:{\(u_{ij}\)}] (u_ij)
         [right=of i] {};
    \node[vertex, label=above:{\(v_{ij}\)}] (v_ij)
         [right=of u_ij] {};
    \node[vertex, fill, label=above:{\(j\)}] (j)
         [right=of v_ij] {};
    \draw (i) -- (u_ij) -- (v_ij) -- (j);
  \end{tikzpicture}
  \caption{Possible stable sets \(S(b)\) extending a given \(b \subseteq V(T)\).
    Black vertices are those which are part of \(S(b)\).}
\label{fig:extension}
\end{figure}
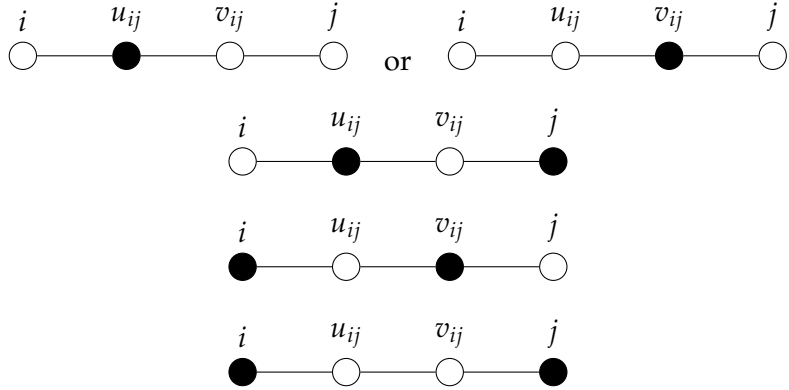

For every subset \(a \subseteq V(T)\), consider the 
induced subgraph \(\gadget{T[a]}\) of \(\gadget{T}\).
By~\eqref{eq:alpha-gadget}, we have \(\alpha(\gadget{T[a]}) 
= \alpha(T[a]) + (\ell + 1) |E(T[a])|\). We also consider 
the induced subgraph \(\smallgadget{T[a]}\) on all the 
\(u_{ij}\) and \(v_{ij}\) in \(\gadget{T[a]}\). This is 
a matching, so obviously \(\alpha(\smallgadget{T[a]}) = 
\size{E(T[a])}\).

By construction, for all sets \(a, b \subseteq V[T]\)
we have:
\begin{align*}
  \size{V(\gadget{T[a]}) \cap S(b)}
  &= \size{S \cap a} + (\ell + 1) \size{E(T[a])} - \size{E(T[a \cap b])},
  \\
  \size{V(\smallgadget{T[a]}) \cap S(b)}
  &= \size{E(T[a])} - \size{E(T[a \cap b])}.
\end{align*}

From the slack matrix \(M\) of \(\STAB(G)\), we construct 
a matrix \(N\) with rows and columns indexed by all 
subsets \(a, b\) of \(V(T)\) with \(a \neq \emptyset\),
with entries
\begin{equation*}
 \begin{split}
  N(a, b)
  &
  \coloneqq
  M(\gadget{T[a]}, S(b)) +
  M(\smallgadget{T[a]}, S(b))
  \\
  &
  =
  \alpha(\gadget{T[a]})
  - \size{V(\gadget{T[a]}) \cap S(b)}
  +
  \alpha(\smallgadget{T[a]})
  - \size{V(\smallgadget{T[a]}) \cap S(b)}
  \\
  &
  =
  \alpha(T[a])
  + 2\size{E(T[a \cap b])} - \size{a \cap b}.
 \end{split}
\end{equation*}
Specializing to \(T = K_{t}\),
we obtain for \(a \neq \emptyset\):
\begin{equation*}
  N(a,b)
  =
  1
  + 2 \binom{\size{a \cap b}}{2} - \size{a \cap b}
  = (1 - \size{a \cap b})^{2}
  =
  \begin{cases}
    1 & \text{if } \size{a \cap b} = 0, \\
    0 & \text{if } \size{a \cap b} = 1.
  \end{cases}
\end{equation*}
Thus, \(N\) contains \(\UDISJ(t)\) as a submatrix
without the row of the empty set.
Now Theorem~\ref{thm:factorization}
followed by Lemma~\ref{lem:nnegrk-preserve} implies
\[
\xc(\STAB(G)) \geqslant \nnegrk M - 1 \geqslant \nnegrk N - 1
\geqslant \nnegrk \UDISJ[t] - 2 \geqslant 2^{t \log (3 / 2)} - 2.
\]
\end{proof}
\end{lem}

\subsection{Existence of gadgets in random graphs}
\label{sec:existence-gadget}

In this section, we estimate the probability that a 
random Erdős–Rényi graph $G = G(n,p)$ contains
an induced copy of a graph $H$.  Recall that in the
$G(n,p)$ model, each of the $\binom{n}{2}$ pairs of 
vertices is connected by an edge with probability $p$,
independently from the other edges. The next lemma 
is key for proving lower bounds on the formulation
complexity of $\STAB(G(n,p))$ via embedding $H = 
\gadget{T}$ as an induced subgraph.
The lemma is formulated in a general for future applications to
many types of subgraphs \(H\).

\begin{lem}
  \label{lem:whpExistence}
  Let \(H\) be a graph with \(v\) vertices and
  with all induced subgraphs having average degree at most~\(d\).
  Let \(0 < p \leqslant 1/2\) and
  \[
  g = g(n,p,v) \coloneqq \frac{v^{2} p^{-\frac{d}{2}} (1-p)^{-\frac{v}{2}}}{n-v}.
  \] 
  The probability of \(G(n, p)\) not containing
  an induced copy of \(H\) satisfies
   \begin{equation*}
    \probability{H \notindsubgraph G(n, p)}
    \leqslant
    c_{0}
      g^2
    \approx 1.23 g^{2}
    ,
  \end{equation*}
  where \(c_{0} \coloneqq \exp ( 2 W ( 1 / \sqrt{2}) ) / 2 \)
  and \(W\) is the Lambert \(W\)-function, the inverse of
  \(x \to x \exp x\).
\end{lem}
\begin{proof}
The proof is via the second-moment method.

Let \(S\) be any graph isomorphic to \(H\) with
\(V(S) \subseteq V(G)\).
Let \(X_{S}\) be the indicator random variable of \(S\) being an
induced subgraph of \(G\).
Obviously,
the total number \(X\) of induced subgraphs of \(G\)
isomorphic to \(H\)
satisfies \(X = \sum_{S} X_{S}\).
We estimate the expectation and variance of \(X\).
Let \(e\) denote the number of edges of \(H\),
and let \(\Aut(H)\) denote the automorphism group of \(H\).
The expectation is clearly
\begin{equation*}
  \expectation{X}
  = \sum_{S} \expectation{X_{S}}
  = \binom{n}{v} \frac{v!}{\card{\Aut(H)}} p^e (1-p)^{\binom{v}{2}-e}.
\end{equation*}
The variance needs more preparations.
Let now \(S\) and \(T\) be two graphs
isomorphic to $H$ with $V(S), V(T) \subseteq V(G)$.
Using that
\(X_{S}\) and \(X_{T}\) are independent and thus
\(\covariance{X_{S}}{X_{T}} = 0\) when 
\(\size{V(S) \cap V(T)} \leqslant 1\) we get
\begin{multline*}
  \variance{X}
  = \sum_{S, T} \covariance{X_{S}}{X_{T}}
  \leqslant
  \sum_{\size{V(S) \cap V(T)} \geqslant 2}
  \expectation{X_{S} X_{T}}
  \\
  =
  \sum_{\size{V(S) \cap V(T)} \geqslant 2}
  \expectation{X_{S}}
  \expectation[X_{S} = 1]{X_{T}}
  =
  \expectation{X}
  \sum_{T \colon \size{V(S) \cap V(T)} \geqslant 2}
  \expectation[X_{S} = 1]{X_{T}}.
\end{multline*}
Note that in the last sum \(S\) is fixed,
and by symmetry, the sum is independent of the actual value of \(S\).
That is why we could factor it out.
We obtain via Chebyshev's inequality,
\begin{equation*}
  \begin{split}
    \probability{H \notindsubgraph G(n, p)} = \probability{X = 0}
    &\leqslant \frac{\variance{X}}{\expectation{X}^2}
    \leqslant
    \frac{\sum_{T \colon \card{V(S) \cap V(T)} \geqslant 2}
      \expectation[X_{S} = 1]{X_{T}}}{\expectation{X}}
    .
  \end{split}
\end{equation*}
We shall estimate \(\expectation[X_{S} = 1]{X_{T}}\),
which is the probability that \(H\) is induced in $G$
provided \(S\) is induced in $G$,
as a function of $k \coloneqq \size{V(S) \cap V(T)}$.
We assume that \(S\) and \(T\) coincide on \(V(S) \cap V(T)\),
and therefore have at most \(dk/2\) edges in common,
as their intersection is isomorphic to an induced subgraph of \(H\),
and therefore have average degree at most \(d\) by assumption.
Hence as \(p \leqslant 1/2\)
\begin{equation*}
  \expectation[X_{S} = 1]{X_{T}}
  = \probability[S \indsubgraph G]{T \indsubgraph G}
  \leqslant
  p^{e - \frac{d}{2}k} (1-p)^{\binom{v}{2} - e - \binom{k}{2} + \frac{d}{2}k}
  .
\end{equation*}
This is clearly also true if \(S\) and \(T\) do not coincide on
\(V(S) \cap V(T)\), as then the probability is \(0\).
Now we can continue our estimation
by summing up for all possible \(T\) with $k \geqslant 2$:
\begin{multline*}
  \frac{\sum_{T}
    \expectation[X_{S} = 1]{X_{T}}}{\expectation{X}}
  \leqslant \frac{\sum_{k=2}^{v}
  \binom{v}{k}
  \binom{n - v}{v - k}
  \frac{v!}{\card{\Aut{H}}}
  p^{e - \frac{d}{2}k} (1-p)^{\binom{v}{2} - e - \binom{k}{2} + \frac{d}{2}k}}%
  {\binom{n}{v} \frac{v!}{\card{\Aut{H}}} p^{e} (1-p)^{\binom{v}{2} - e}}
  \\
  =
  \sum_{k=2}^{v}
  \frac{\binom{v}{k} \binom{n - v}{v - k}}{\binom{n}{v}}
  p^{- \frac{d}{2} k} \left( \underbrace{(1-p)^{\frac{d + 1 - k}{2}}}_{\leqslant (1-p)^{-\frac{v}{2}}} \right)^{k}
  \leqslant
  \sum_{k=2}^{v}
  \frac{v^{k}}{2 (k-2)!}
  \left(\frac{v}{n-v}\right)^k
  \left(p^{-\frac{d}{2}} (1-p)^{-\frac{v}{2}}\right)^{k}\\
  =
  \frac{1}{2}
  g^2 
  \sum_{k=2}^{v}
  \frac{1}{(k-2)!}
  g^{k-2}
  \leqslant
  \frac{1}{2}
  g^2
  \exp(g),
\end{multline*}
as
\begin{equation*}
  \frac{\binom{v}{k} \binom{n-v}{v-k}}{\binom{n}{v}}
  \leqslant
  \frac{\binom{v}{k} \frac{(n-v)^{v-k}}{(v-k)!}}{\frac{(n-v)^{v}}{v!}}
  =
  \binom{v}{k}^{2} \frac{k!}{(n-v)^{k}}
  \leqslant \frac{1}{k!} \genfrac(){}{}{v}{n - v}^{k}.
\end{equation*}
The lemma follows: the probability of \(H\)
not being an induced subgraph is at most \(e^{g} g^{2} / 2\).
This upper bound is \(1\) exactly if \(g = 2 W (1 / \sqrt{2})\).
For \(g \leqslant 2 W (1 / \sqrt{2})\),
we obtain the upper bound in the lemma.
For \(g \geqslant 2 W (1 / \sqrt{2})\),
the upper bound in the lemma is at least \(1\),
so the statement is obvious.
\end{proof}

\subsection{High formulation complexity with high probability}
\label{sec:high-extens-compl}

In order to obtain lower bounds on the formulation complexity 
of the maximum stable set problem of $G = G(n,p)$, via Lemmas
\ref{lem:whpExistence} and \ref{lem:xc-gadget},
taking $H$ to be $\gadget{K_t}$.
We obtain the following result:

\begin{mainthm}[Super-polynomial xc of \(\STAB(G(n,p))\) w.h.p.]
  \label{mainthm:whp-xc-STAB-LB}
  With high probability,
  the maximum stable set problem of the random graph \(G(n, p)\)
  has at least the following formulation complexity,
  depending on the size of \(p\):
  \begin{enumerate}
  \item\label{item:STAB-xc-high-regime}
  For \(p = \omega( 1 / \sqrt[4]{n} )\) and
  fixed
  \(0 < c < 2 / \sqrt{3} \approx 1.1547\), 
  we have
  \begin{equation}
    \label{eq:non-uni-xc-general}
  \probability{\xc(\STAB(G(n, p)))
    \geqslant 2^{\sqrt{c \frac{\ln (n p^{4})}{p}} \log(3/2)}}
  = 1 - o(1)
  .
  \end{equation}
  \item\label{item:STAB-xc-middle-regime}
  For \(c > 0\) and \(c / \sqrt[3]{n} \leqslant p = o(1)\)
  we have
  \begin{equation}
    \label{eq:non-uni-xc-1/3-1/4}
    \probability{\xc(\STAB(G(n, p)))
      \geqslant 2^{\frac{\log(3/2)}{\sqrt{p \ln (1/p)}}}}
    = 1 - O(1/c^{6})
    .
  \end{equation}
  \item\label{item:STAB-xc-low-regime}
  Moreover, for any fixed \(c > 0\)
  for all \(1 / n < p \leqslant  c / \sqrt[3]{n}\)
  and \(0 < \delta < 1\)
  \begin{equation}
    \label{eq:non-uni-xc-less-1/3}
    \probability{\xc(\STAB(G(n, p)))
      \geqslant 2^{\delta \sqrt{\frac{\sqrt{pn}}{\ln (1/p)}} \log (3/2)}}
    \geqslant 1 - O(\delta^{8})
    .
  \end{equation}
  \end{enumerate}
\end{mainthm}

As an illustration of Main theorem~\ref{mainthm:whp-xc-STAB-LB},
we include concrete lower bounds in special cases of interest. 

\begin{cor} \label{cor:LBs}
  For every fixed \(0 < \varepsilon < 1\),
  we have
  \begin{align}
    \label{eq:non-uni-xc-n-1/4}
    \probability{\xc(\STAB(G(n, n^{-\varepsilon})))
      \geqslant 2^{\sqrt{(1 - 4 \varepsilon) n^\varepsilon \ln n} \log(3/2)}}
    &
    = 1 - o(1)
    &
    \text{for } \varepsilon &< 1/4
    ,
    \\
    \label{eq:non-uni-xc-n-1/3-1/4}
    \probability{\xc(\STAB(G(n, n^{- \varepsilon})))
      \geqslant
      2^{\frac{n^{\varepsilon / 2}}{\sqrt{\varepsilon \ln n}}
        \log(3/2)}}
    &
    = 1 - o(1)
    &
    \text{for } \varepsilon &< 1/3
    ,
    \\
    \label{eq:non-uni-xc-n-1-eps}
    \probability{\xc(\STAB(G(n, n^{- \varepsilon})))
      \geqslant
      2^{\frac{n^{(1 - \varepsilon) / 4}}{\ln n}\log (3/2)}}
    &
    = 1 - o(1)
    &
    \text{for } \varepsilon &\geqslant 1/3
    .
  \end{align}
  Below the \(p = n^{- \varepsilon}\)
  range, we obtain
  \begin{equation}
    \label{eq:non-uni-xc-polylog/n}
    \probability{\xc(\STAB(G(n, (\ln^{6 + \varepsilon} n) / n)))
      \geqslant 2^{\ln^{1 + \varepsilon / 5} n \cdot \log (3/2)}}
    = 1 - o(1)
    ,
  \end{equation}
  and (at the other end of the range) for fixed \(\delta > 0\),
  \begin{equation}
    \label{eq:non-uni-xc-ln/n}
    \probability{\xc(\STAB(G(n, \delta \ln^{-1} n)))
      \geqslant n^{\delta^{- 1/2} \log(3/2)}}
    = 1 - o(1)
    .
  \end{equation}
\begin{proof}
Equations~\eqref{eq:non-uni-xc-n-1/4} and \eqref{eq:non-uni-xc-ln/n}
are special cases of \eqref{eq:non-uni-xc-general}.
For Equation~\eqref{eq:non-uni-xc-n-1/4},
we choose \(p = n^{-\varepsilon }\) and \(c = 1\).
For Equation~\eqref{eq:non-uni-xc-general},
we choose \(p = \delta \ln^{-1} n\) and
\(c = 1.1\), a bit larger than \(1\),
then the square root in \eqref{eq:non-uni-xc-general}
becomes
\begin{equation*}
  \sqrt{c \frac{\ln n p^{4}}{p}}
  =
  \sqrt{c \frac{\ln n + 4 \ln (\delta) + 4 \ln \ln^{-1} n}
    {\delta} \ln n}
  = c \delta^{-1/2} (1 + o(1)) \ln n
  > \delta^{-1/2} \ln n,
\end{equation*}
proving the equation.

Equation~\eqref{eq:non-uni-xc-n-1/3-1/4}
follows from Equation~\eqref{eq:non-uni-xc-1/3-1/4}
via \(p = n^{- \varepsilon}\). Equations~\eqref{eq:non-uni-xc-n-1-eps} and
\eqref{eq:non-uni-xc-polylog/n}
are special cases of Equation~\eqref{eq:non-uni-xc-less-1/3}.
Equation~\eqref{eq:non-uni-xc-n-1-eps} is the case
\(p = n^{- 1 + \varepsilon}\)
and \(\delta = \sqrt{\varepsilon} \ln^{-1} n\).
For Equation~\eqref{eq:non-uni-xc-polylog/n},
we choose
\(p = (\ln^{6 + \varepsilon} n) / n\)
and
\(\delta = \ln^{- \varepsilon / 20} n\),
then the interesting part of the exponent is
\begin{equation*}
  \delta \sqrt{\frac{\sqrt{pn}}{\ln (1/p)}}
  =
  \ln^{- \varepsilon / 20} n
  \sqrt{\frac{\sqrt{\ln^{6 + \varepsilon} n}}
    {\ln n - \ln \ln^{6 + \varepsilon} n}}
  >
  \ln^{- \varepsilon / 20} n
  \sqrt{\frac{\sqrt{\ln^{6 + \varepsilon} n}}
    {\ln n}}
  = \ln^{1 + \varepsilon / 5} n
\end{equation*}
proving the claim.
\end{proof}
\end{cor}

Now we are going to prove the main theorem
of Section~\ref{sec:high-extens-compl}.

\begin{proof}[Proof of Main Theorem~\ref{mainthm:whp-xc-STAB-LB}]
We apply Lemma \ref{lem:whpExistence}
to the graph $H \coloneqq \gadget{K_t}$
together with Lemma \ref{lem:xc-gadget} to obtain: 
\begin{equation}
  \label{eq:non-uni-xc-main}
  \begin{split}
  \probability{\xc(\STAB(G(n, p)))
    \geqslant 2^{t \log (3/2)}}
  &\geqslant
  \probability{\gadget{K_{t}} \indsubgraph G(n, p)}\\
  &\geqslant
  1 - c_{0} \,
  \frac{v^{4} p^{-d} (1-p)^{-v}}{(n-v)^{2}}\\
  &\geqslant
  1 - c_{0} \,
  (1+ o(1))
  \frac{v^{4} p^{-d} \mathrm{e}^{p v}}{n^{2}}
  \qquad
  \text{if \(v = o(n)\)}
  .
  \end{split}
\end{equation}
Here \(v\) is the number of vertices of \(H\),
and every induced subgraph of \(H\)
should have average degree at most \(d\).
We shall estimate the last fraction
\(v^{4} p^{-d} {\mathrm{e}}^{pv} / n^{2}\),
using the \(d\) provided by Lemma~\ref{lem:avg-deg-gadget}.
Below we will tacitly assume \(t = \omega(1)\), which is 
w.l.o.g\ because \(\xc(\STAB(G(n, p))) \geqslant n\)
always.

Now we shall substitute various values for \(p, t, d, \ell\)
to obtain the equations of the theorem.
We will verify \(v = o(n)\)
and \(v^{4} p^{-d} \mathrm{e}^{p v} / n^{2} = o(1)\)
to obtain an \(1-o(1)\) lower bound from the last inequality.

For establishing~\eqref{eq:non-uni-xc-general}, we choose
\begin{align*}
\ell &\coloneqq 0 &
t &\coloneqq \roundup*{c \sqrt{\frac{\ln (n p^{4})}{p}}}
&d &\coloneqq 4.
\end{align*}
Note that for \(p \geqslant 1 / \sqrt[4]{n}\),
\begin{equation*}
  v = t + 3 \binom{t}{2} = \left(\frac{3}{2} + o(1)\right) t^{2}
  = \left(\frac{3}{2} + o(1)\right) c^{2} \frac{\ln (n p^{4})}{p}
  \leqslant
  \left(\frac{3}{2} + o(1)\right) c^{2} \sqrt[4]{n} \ln n = o(n),
\end{equation*}
and hence
\begin{equation*}
 \begin{split}
  \frac{v^{4} p^{-d} \mathrm{e}^{p v}}{n^{2}}
  &
  =
  \left( \frac{3}{2} + o(1) \right)^{4}
  (p t^{2})^{4} \mathrm{e}^{(3/2 + o(1)) p t^{2}
    - 2 \ln (n p^{4})}
  \\
  &
  \leqslant
  \left( \frac{3}{2} + o(1) \right)^{4} c^{8}
  \left( \ln (n p^{4}) \right)^{4}
  \exp\left\{
    \left[
      \left(
        \frac{3}{2} + o(1)
      \right) c^{2}
      - 2
    \right]
    \ln (n p^{4})
  \right\}
  = o(1)
  ,
 \end{split}
\end{equation*}
as \(n p^{4} = \omega(1)\) by assumption.
This finishes the proof of~\eqref{eq:non-uni-xc-general}.

We turn to~\eqref{eq:non-uni-xc-1/3-1/4} and
\eqref{eq:non-uni-xc-less-1/3}. We will choose
a positive \(\ell\) to approximately minimize
the fraction in terms of the other parameters.
To ease computation, let
\[
\gamma \coloneqq \frac{2 \ell + 3}{2} > 1
.
\]
Then the parameters \(v\) and \(d\) look like
\begin{align*}
  d &= 2 + \frac{4}{2 \ell + 3} = 2 + \frac{2}{\gamma},
  \\
  v &= t + (2 \ell + 3) \binom{t}{2}
  = \gamma t^{2} + (1 - \gamma) t < \gamma t^{2}.
\end{align*}
Hence
\begin{equation*}
  \frac{v^{4} p^{-d} {\mathrm{e}}^{p v}}{n^{2}}
  <
  \frac{\gamma^{4} t^{8}
    {\mathrm{e}}^{p \gamma t^{2} + 2(\ln (1/p)) / \gamma}}
  {p^{2} n^{2}}
  .
\end{equation*}
The \(\gamma\) minimizing the expression is
\begin{equation*}
  \frac{\sqrt{4 + 2 p t^{2} \ln (1/p)} - 2}{p t^{2}}
  =
  \frac{2 \ln (1/p)}{\sqrt{4 + 2 p t^{2} \ln (1/p)} + 2},
\end{equation*}
but we use an approximation as \(\ell\) needs to be an even integer.
Therefore we choose
\begin{equation*}
  \ell = 2 \roundup*{
    \frac{\ln (1/p)}{\sqrt{4 + 2 p t^{2} \ln (1/p)} + 2}
    - \frac{3}{4}
  }
  .
\end{equation*}
We will verify later that actually \(\ell = \omega (1)\).  Hence
\begin{equation*}
  \gamma = (1 + o(1))
  \frac{2 \ln (1/p)}{\sqrt{4 + 2 p t^{2} \ln (1/p)} + 2}
  =
  (1 + o(1))
  \frac{\sqrt{4 + 2 p t^{2} \ln (1/p)} - 2}{p t^{2}}
  ,
\end{equation*}
and
\begin{equation}
  \label{eq:mainthm-prob-est}
 \begin{split}
  \frac{v^{4} p^{-d} {\mathrm{e}}^{p v}}{n^{2}}
  &
  < (1 + o(1))
  \left(
    \frac{2 p t^{2} \ln (1/p)}{\sqrt{n p^{3}}}
  \right)^{4}
  \frac{{\mathrm{e}}^{(2 + o(1)) \sqrt{4 + 2 p t^{2} \ln (1/p)}}}
  {\left(
      \sqrt{4 + 2 p t^{2} \ln (1/p)} + 2
    \right)^{4}}
  \\
  &
  =
  (1 + o(1))
  \frac{{\mathrm{e}}^{(2 + o(1)) \sqrt{4 + 2 p t^{2} \ln (1/p)}}
    \left(
      \sqrt{4 + 2 p t^{2} \ln (1/p)} - 2
    \right)^{4}
  }
  {\left(
      n p^{3}
    \right)^{2}}
  .
 \end{split}
\end{equation}

Now we shall substitute various values for \(p\) and \(t\)
to obtain the equations of the theorem.

We will need to verify \(\ell = \omega(1)\) and \(v = o(n)\)
for every choice.

For Equation~\eqref{eq:non-uni-xc-less-1/3},
i.e., in the case \(1 / n < p \leqslant  c / \sqrt[3]{n}\),
we neglect the exponential term in \eqref{eq:mainthm-prob-est}
for the choice of \(t\):
\begin{equation*}
  t = \roundup*{\delta \sqrt{\frac{\sqrt{pn}}{\ln (1/p)}}}
  .
\end{equation*}
Here \(0 < \delta < 1\) is an additional parameter.
Rearranging gives us
\begin{equation*}
  2 p t^{2} \ln (1/p) = (1 + o(1)) \delta^{2} \sqrt{n p^{3}}
 \leqslant (1 + o(1)) \delta^{2} c^{3/2} \leqslant (1 + o(1)) c^{3/2},
\end{equation*}
so in particular,
\begin{gather*}
  \ell \geqslant
  2
  \roundup*{\frac{\ln (\sqrt[3]{n} / c)}
    {\sqrt{4 + (1 + o(1)) c^{3/2}} + 2}
    - \frac{3}{4}}
  = \omega(1)
  \\
  v < \gamma t^{2} = O(1/p)
  = O(\sqrt[3]{n}) = o(n)
  .
\end{gather*}
Finally,
\begin{equation*}
 \begin{split}
  \frac{v^{4} p^{-d} {\mathrm{e}}^{p v}}{n^{2}}
  &
  < (1 + o(1))
  \frac{{\mathrm{e}}^{(2 + o(1)) \sqrt{4 + (1 + o(1)) c^{3/2}}}
    \left(
      \sqrt{4 + (1 + o(1)) \delta^{2} \sqrt{n p^{3}}} - 2
    \right)^{4}
  }
  {\left(
      n p^{3}
    \right)^{2}}
  \\
  &
  \leqslant
  (1 + o(1))
  {\mathrm{e}}^{(2 + o(1)) \sqrt{4 + (1 + o(1)) c^{3/2}}}
    \left(
      (1/4 + o(1)) \delta^{2}
    \right)^{4}
  = O(\delta^{8})
  ,
 \end{split}
\end{equation*}
as claimed.

For Equation~\eqref{eq:non-uni-xc-1/3-1/4},
i.e., when \(c / \sqrt[3]{n} \leqslant p = o(1)\),
we choose
\begin{equation*}
  t = \roundup*{\frac{1}{\sqrt{p \ln (1/p)}}}
  .
\end{equation*}
This provides the estimate
\begin{equation*}
  2 p t^{2} \ln (1/p) = 2 + o(1)
  ,
\end{equation*}
hence \(\ell = \Theta(\ln (1/p)) = \omega(1)\),
and \(v < \gamma t^{2} = O(1/p) = O(\sqrt[3]{n}) = o(n)\).
Finally,
\begin{equation*}
 \begin{split}
  \frac{v^{4} p^{-d} {\mathrm{e}}^{p v}}{n^{2}}
  &
  =
  (1 + o(1))
  \frac{{\mathrm{e}}^{(2 + o(1)) \sqrt{4 + (2 + o(1))}}
    \left(
      \sqrt{4 + (2 + o(1))} - 2
    \right)^{4}
  }
  {\left(
      n p^{3}
    \right)^{2}}
  \\
  &
  = O \left( \frac{1}{(n p^{3})^{2}} \right)
  = O(1 / c^{6})
  ,
 \end{split}
\end{equation*}
as \(n p^{3} \geqslant c^{3}\).
\end{proof}

Main Theorem~\ref{mainthm:whp-xc-STAB-LB} gives 
super-polynomial lower bounds all the way from 
$p = \Omega(\frac{\log^{6+\varepsilon} n}{n})$ to $p = O(\frac{1}{\log n})$. The key for being able to cover the whole regime is to have the
gadgets depend on the parameter choice. Notice that for $p < 1/n$ a 
random graph almost surely will have all its components of size 
$O(\log n)$, making the stable set problem easy to solve, so that 
we essentially leave only a small polylog gap.

\subsection{Upper bound on formulation complexity with high probability}
\label{sec:ub-extens-compl}

We now complement Main Theorem~\ref{mainthm:whp-xc-STAB-LB}
with an upper bound, which is close to the lower bound, up to an
essentially quadratic gap in the exponent.

\begin{thm}[Upper bound on the xc of \(\STAB(G(n,p))\) w.h.p.]
  \label{thm:xc-STAB-less-whp}
  For \(0 < p \leqslant 1/2\),
  \begin{equation}
    \label{eq:xc-STAB-upper}
  \probability{\xc(\STAB(G))
    \geqslant 2^{\Omega \left( \frac{\ln^2 n}{p} \right)}}
  \leqslant n^{- \Omega \left( \frac{\ln n}{p} \right)}.
  \end{equation}
  In particular, for \(p = n^{-\varepsilon}\),
  we obtain
  \(\probability{\xc(\STAB(G))
    \geqslant 2^{\Omega \left(  n^{\varepsilon} \ln^2 n \right)}}
  = o(1)\)
  and similarly for  $p = \delta \ln^{-1} n$,
  we get
  \(\probability{\xc(\STAB(G))
    \geqslant n^{\Omega \left(\frac{\ln^{3} n}{\delta}\right)}}
  = o(1)\).
\end{thm}

The upper bound stated in Theorem~\ref{thm:xc-STAB-less-whp}
is actually an upper bound on the number of stable sets in \(G\),
i.e., follows from \eqref{eq:rkFromVertex}.

\begin{proof}[Proof of Theorem~\ref{thm:xc-STAB-less-whp}]
By standard arguments
(see, e.g., \cite[Chapter 11, page 300]{Diestel}),
for $G = G(n,p)$ we have
\[
\probability{\alpha(G) \geqslant r}
\leqslant \left(n \, \mathrm{e}^{-p(r-1)/2}\right)^r
\]
and thus for $r = 4 \frac{\ln n}{p}$ we get
\[
\probability{\alpha(G) \geqslant 4 \frac{\ln n}{p}}
\leqslant \left(\frac{n}{\sqrt{e}}\right)^{-4\frac{\ln n}{p}}.
\]
Therefore, with very high probability, we have 
$\alpha(G) \leqslant 4 \frac{\ln n}{p}$. Using the inequality
$\sum_{i=0}^k \binom{n}{i} \leqslant (n+1)^k$,
we get
\[
\#\textrm{(stable sets in $G$)} \leqslant
(n+1)^{\alpha(G)} = 2^{\log(n+1) \alpha(G)} = 2^{\left(\frac{1}{\ln 2}+o(1)\right) \ln (n) \alpha(G)}.
\]
The result then follows directly from \eqref{eq:rkFromVertex}.
\end{proof}

\section{Concluding remarks}
\label{sec:concluding-remarks}

We conclude with the following conjecture whose validity, 
we believe, is necessary to strengthen the result, close the 
remaining gap, as well as establishing truly exponential 
lower bounds on the extension complexity of further 
combinatorial problems. 

\begin{conjecture}[Sparse Graph Conjecture] There exists an infinite
family \((T_k)_{k \in \N}\) of template graphs such that, denoting 
by $t_k$ the number of vertices of $T_k$:
\begin{inparaenum}[(i)]
\item
  \(\xc(\STAB(\gadget{T_k})) = 2^{\Omega(t_k)}\);
\item
  $T_k$ has bounded average degree;
\item
  $t_{k} \leqslant t_{k+1}$ but at the same time $t_{k+1} = O(t_k)$.
\end{inparaenum}
\end{conjecture}

The existence of such a family would have various consequences. 

\paragraph{Exact case.} Assuming the Sparse Graph Conjecture we 
would obtain that the extension complexity of polytopes (see, e.g., 
\cite{extform4,FMPTW15jour} for definitions) for important
combinatorial problems considered in~\cite{extform4,FMPTW15jour,AvisTiwary15,%
VP2013} including (among others) the stable set polytope, knapsack
polytope, and the 3SAT polytope would have truly exponential extension
complexity, that is $2^{\Omega(n)}$ extension complexity, where $n$ 
is the \emph{dimension} of the polytope. 

The recent groundbreaking result of \cite{Rothvoss14} gives
$2^{\Omega(n)}$ bounds for the extension complexity of the
matching polytope and TSP polytope. These bounds are also tight
up to constants, but this time the upper bound does not come from
the number of vertices but rather from the number of facets and 
dynamic programming algorithms, respectively. Notice that the
dimension of both polytopes is $d = \Theta(n^2)$, thus the bounds
are in fact $2^{\Omega(\sqrt{d})}$.

\paragraph{Average case.} As observed above, there is a quadratic 
gap in the best current lower and upper bounds on the worst-case 
extension complexity of the stable set polytope: 
\(2^{\Omega(\sqrt{n})}\) versus \(2^n\) respectively. This is 
reflected in the results we obtain here. Assuming the Sparse Graph
Conjecture we could reduce the gap between upper and lower bounds 
to a logarithmic factor. Moreover, our results could be strengthened 
to establish super-polynomial lower bounds on the average-case 
formulation complexity of $\STAB(G(n,p))$ up to constant probability 
\(p\). 

\section*{Acknowledgements}
\label{sec:acknowledgements}

Research reported in this paper was partially supported
by NSF grant CMMI-1300144.

\bibliographystyle{abbrvnat}
\bibliography{bibs}

\end{document}